\documentclass[superscriptaddress,amsmath,A4,pra,twocolumn,showpacs]{revtex4} 
\usepackage{amssymb}
\usepackage{amsthm}
\usepackage{graphicx}
\usepackage{color}
\usepackage{young}

\input{qcircuit}

\newcommand{\bra}[1]{\langle#1|}
\newcommand{\ket}[1]{|#1\rangle}
\newcommand{\braket}[2]{\langle#1|#2\rangle}
\newcommand{\ketbra}[2]{{\ket{#1}\bra{#2}}}
\newcommand{\Bra}[1]{\langle \! \langle#1|}
\newcommand{\Ket}[1]{|#1\rangle \! \rangle}
\newcommand{\BraKet}[2]{\langle \! \langle #1|#2 \rangle \! \rangle}
\newcommand{\KetBra}[2]{{\Ket{#1}\Bra{#2}}}

\newcommand{\hilb}[1]{\mathcal{#1}}

\def\<{\langle}\def\>{\rangle}
\DeclareMathOperator{\Tr}{Tr}

\DeclareMathOperator{\spn}{span}

\newcommand{\chu}{\mathcal{U}}

\newcommand{\usbl}{C}
\newcommand{\DP}{\widetilde{d}}
\newcommand{\corn}{\textrm{Cor}}
\newcommand{\jJall}{J^{\Box}}

\newtheorem{lemma}{Lemma}
\newtheorem{theorem}{Theorem}

\begin{document}

\title{Optimal probabilistic storage and retrieval of unitary channels}

\author{Michal Sedl\'ak}
\affiliation{RCQI, Institute of Physics, Slovak Academy of Sciences, D\'ubravsk\'a cesta 9, 84511 Bratislava, Slovakia}
\affiliation{Faculty of Informatics,~Masaryk University,~Botanick\'a 68a,~60200 Brno,~Czech Republic}
\author{Alessandro Bisio}
\affiliation{QUIT group, Dipartimento di Fisica, INFN Sezione di Pavia, via Bassi
  6, 27100 Pavia, Italy}
\author{M\'ario Ziman}
\affiliation{RCQI, Institute of Physics, Slovak Academy of Sciences, D\'ubravsk\'a cesta 9, 84511 Bratislava, Slovakia}
\affiliation{Faculty of Informatics,~Masaryk University,~Botanick\'a 68a,~60200 Brno,~Czech Republic}

\begin{abstract}
 We address the question of a quantum memory storage of
  quantum dynamics. In particular, we design an optimal
  protocol for $N\to 1$ probabilistic storage-and-retrieval of unitary
  channels on $d$-dimensional quantum systems.
  If we may access the unknown unitary gate only $N$-times, the
  optimal success probability of perfect retrieval of its single use
  is $N/(N-1+d^2)$. The derived size of the memory system
  exponentially improves the known upper bound on the size of the program
  register needed for probabilistic programmable quantum processors.
  Our results are closely related to probabilistic perfect alignment
  of reference frames and probabilistic port-based teleportation.
 \end{abstract}



\maketitle

\emph{Introduction.}  Since the discovery of the first quantum algorithms
\cite{shor,grover} and protocols \cite{bb84,teleportation} the
information processing with quantum systems has challenged
basic paradigms and existing limitations of computer science. In the
last few decades we have discovered that quantum information cannot be
cloned \cite{noclon}, its ``logical value" cannot be inverted
\cite{not_gate}, quantum processors cannot be universally programmed
\cite{nielsen1}, and universal multimeters do not exist
\cite{multimeter,multimeter1}. No doubt, any of these programmable devices would
represent a very useful piece of quantum technology, thus,
their approximate realisations are of foundational interest
\cite{buzek,multimeter1,scarani,processor,multimeter}. The no-go
restrictions imposed by quantum theory are treated in two ways. Either
we ask for an approximate performance, or we allow that the perfect
performance happens with some probability of failure.


Studies of optimal approximate cloners initiated by Hillery and Bu\v
zek \cite{buzek} demonstrated that such non-ideal devices are of practical
relevance and this motivated the study of other universal devices.
In particular, it was shown that
quantum theory limits the fidelity of $1\to N$ clones of qubits to
$(2N+1)/3N$ \cite{gisin1997}. For quantum processors Nielsen and
Chuang \cite{nielsen1} proved that perfect (error free) implementation
of $k$ distinct unitary transformations requires at least $k$
dimensional program register. Recently, the cloning was considered
also for quantum transformations \cite{chiri1,gtuc}. This unveiled an unexpected
feature called super-replication \cite{dur,chiri2}. In this protocol,
starting with $N$ copies of a qubit unitary transformation $U$
one deterministically generates up to $N^2$ copies of $U$
with an exponentially small error rate.
While studying cloning of unitaries it was realized there is a closely
related task of storage-and-retrieval (SAR), which only differs in the
causal order of available resources.
While in the cloning the cloned device is available after the input
states are at the disposal, one can consider also a task where this
order is reversed, thus, the device is available only before the
input states. In such case, we need to learn \cite{foot1}
and somehow \emph{store} the action of the device and \emph{retrieve}
it once the input states are available.

\emph{Problem formulation.}  The devices transforming states of
$d$-dimensional quantum systems associated with Hilbert space
$\hilb{H}$ are formalized as quantum channels, i.e. completely
positive trace-preserving linear maps on the space
$\mathcal{L}(\hilb{H})$ of linear operators on $\hilb{H}$. Suppose an
unknown channel $\chu$ is provided for experiments and we may access it $N$
times. However, we are asked to apply $\chu$ on an unknown state $\xi$
only after we lost the access to this channel. Therefore, our aim is
to find an optimal strategy that stores $\chu$ in a state of a quantum
memory (associated with Hilbert space $\hilb{H}_M$) and allows us to
retrieve its action when needed. In the approximative settings this
task (for unitary channels) was studied in Ref.\cite{bisilearn}.
\begin{figure}
  \begin{center}
    \includegraphics[width=8cm]{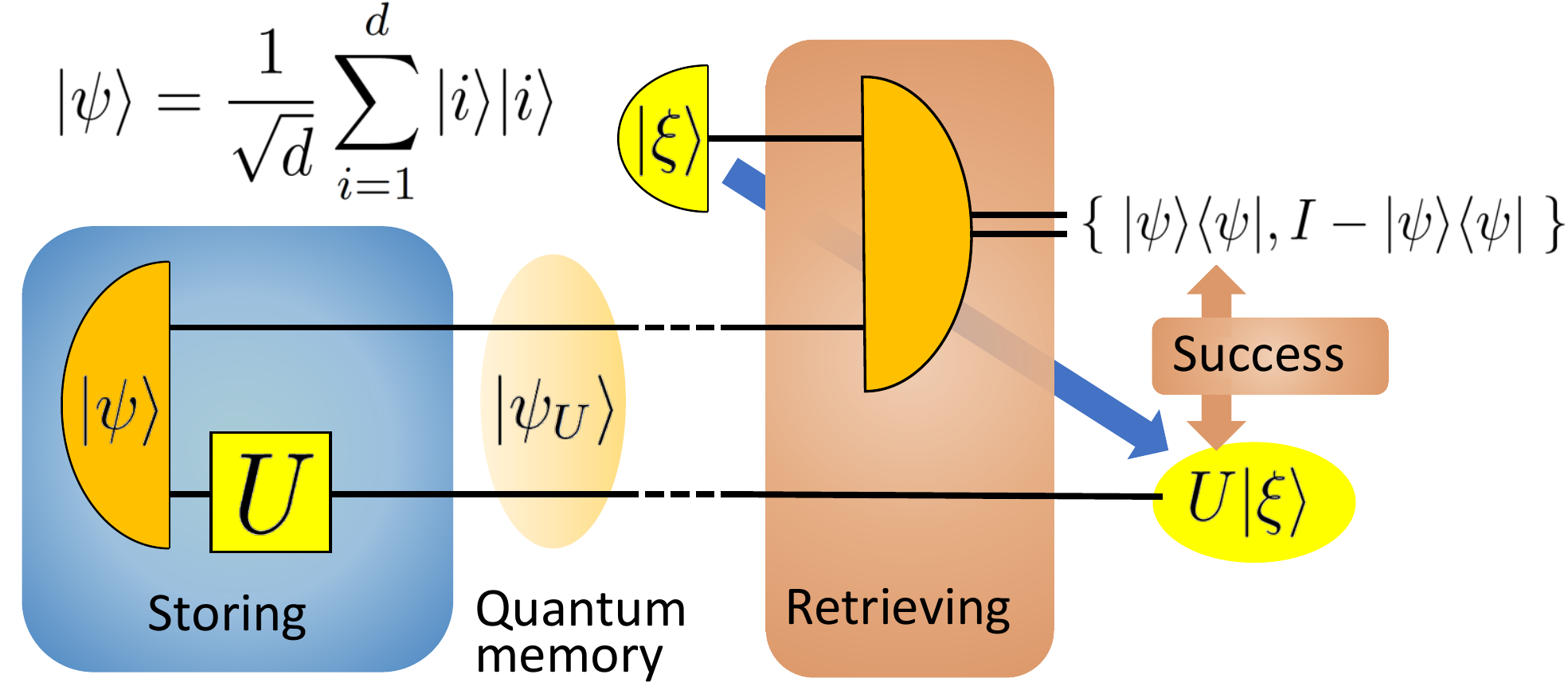}
\caption{Optimal $1\rightarrow 1$ PSAR of unitary channels.}
\label{fig:bs1}
  \end{center}
\end{figure}

Our goal is to investigate the probabilistic version of the SAR
problem, in particular, we aim to find the optimal $N\rightarrow 1$
probabilistic storage and retrieval procedure (PSAR). Moreover, we
require the retrieved channel to be implemented perfectly and with
the same probability of success (``covariance" property)
for all considered channels. We will design
the strategy maximizing the probability for the set of unitary
channels, i.e. $\chu(\xi)=U \xi U^\dagger$ for some unitary operator
$U$. Due to no-programming theorem \cite{nielsen1}, the
retrieving part of any PSAR strategy cannot  be deterministic. Thus, the
successful retrieval is described by a trace-non-increasing
completely positive linear map (quantum operation)
$\mathcal{T}_U:\mathcal{L}(\hilb{H})\to\mathcal{L}(\hilb{H})$
proportional to the unknown unitary channel, $\mathcal{T}_U = \lambda_U
\chu$. Consequently, the success probability is $\lambda_U={\rm
  tr}[\mathcal{T}_U(\xi)]$
and the condition of covariance
implies $\lambda_U=\lambda$ for all $U$.

\emph{One-to-one probabilistic storage-and-retrieval.}  In
such case the unknown unitary $U$ is applied on a suitably chosen
state $\ket{\psi}$ (in general bipartite and entangled), which yields
state $\ket{\psi_U}\in\hilb{H}_M$ and concludes the storing
phase. Afterwards, once we want to apply unitary $U$ on some state
$\xi$, we employ a retrieving quantum instrument
$\mathbf{R}=\{\mathcal{R}_s, \mathcal{R}_f \}$, which acts on
$\xi\otimes\ket{\psi_U}\bra{\psi_U}$ and in case of success outputs an
sub-normalized state $\lambda U\xi U^\dagger$,
i.e.
$\mathcal{R}_s:\mathcal{L}(\hilb{H}_{\rm
  in}\otimes\hilb{H}_M)\to\mathcal{L}(\hilb{H}_{\rm out})$ with
$\hilb{H}=\hilb{H}_{\rm in}=\hilb{H}_{\rm out}$. The retrieving
quantum instrument plays the role of a probabilistic programmable
processor and the state $\ket{\psi_U}$ programs a unitary
transformation $U$ to be performed on a state $\xi$.

Using the Choi isomorphism \cite{choi} we have that
$\mathcal{R}_s(\xi\otimes\ket{\psi_U}\bra{\psi_U})={\rm tr}_{{\rm in},M}[(I\otimes\xi^T\otimes\ket{\psi_U}\bra{\psi_U}^T)R_s]=\lambda{\rm tr}_{\rm in}[(I\otimes\xi^T)\Ket{U}\Bra{U}]=\lambda U\xi U^\dagger$, where
$R_s\in\mathcal{L}(\hilb{H}_{\rm out}\otimes\hilb{H}_{\rm in}\otimes\hilb{H}_M)$
and $\Ket{U}=\sqrt{d}(U\otimes I)\ket{\psi_+}$ with
$\ket{\psi_+}=d^{-1/2}\sum_j \ket{j}\otimes\ket{j}$ (vectors $\{\ket{j}\}$ form
an orthonormal basis of $\hilb{H}=\hilb{H}_{\rm in}=\hilb{H}_{\rm out}$).
Since the above identity must
hold for any $\xi$ and $\ket{\psi_U}\bra{\psi_U}^T=\ket{\psi^*_U}\bra{\psi^*_U}$ (both the transposition and the conjugation are defined with respect to the same basis of $\hilb{H}_M$) we obtain the following \emph{perfect retrieval condition}
\begin{align}
  \label{eq:case11composition}
  \begin{aligned}
    & \bra{\psi^*_U} R_s \ket{\psi^*_U} = \lambda \KetBra{U}{U}
    \quad\forall U\in SU(d)
    \,.
  \end{aligned}
\end{align}

Already this simple case shows that the maximization of
probability of success $\lambda$ involves the simultaneous optimization of
the storing phase (choice of $\ket{\psi}$) and the retrieving phase
(choice of quantum instrument $\mathbf{R}$). It turns out that the optimal
performance is achieved by the (incomplete) quantum teleportation
protocol \cite{teleportation} that is a known example of a universal
probabilistic quantum processor \cite{qbook}. Let us note that this is
similar to quantum gate teleportation invented by Gottesman and
Nielsen \cite{nielsen2}, yet it is different, because
PSAR must work perfectly for any unitary
transformation. In particular,
for the storing phase we set $\ket{\psi}=\ket{\psi_+}$.
Then the optimal retrieval is achieved by a quantum teleportation of state $\xi$ using the stored state $\ket{\psi_U}=d^{-1/2}\Ket{U}$ (see Fig.~\ref{fig:bs1}). The generalized Bell measurement performed on $\xi$ and one part of $\ket{\psi_U}$ results in an outcome $k$ with probability $1/d^2$. In such case we are left with the second part of $\ket{\psi_U}$ in the state $U\sigma_k\xi \sigma_k U^\dagger$, where $\sigma_k$ are generalized Pauli operators. In case of $\sigma_k=I$ (associated with the Bell measurement projection onto $\ket{\psi_+}$) the stored unitary channel is successfully retrieved. For all the other outcomes, the unwanted $\sigma_k$ rotation can not be undone, because the unitary $U$ is unknown. In conclusion, the teleportation-based PSAR succeeds with probability
$1/d^2$. Its optimality follows from our subsequent discussion of the
optimal $N\to 1$ PSAR.

\begin{figure}
  \begin{center}
    \includegraphics[width=8cm]{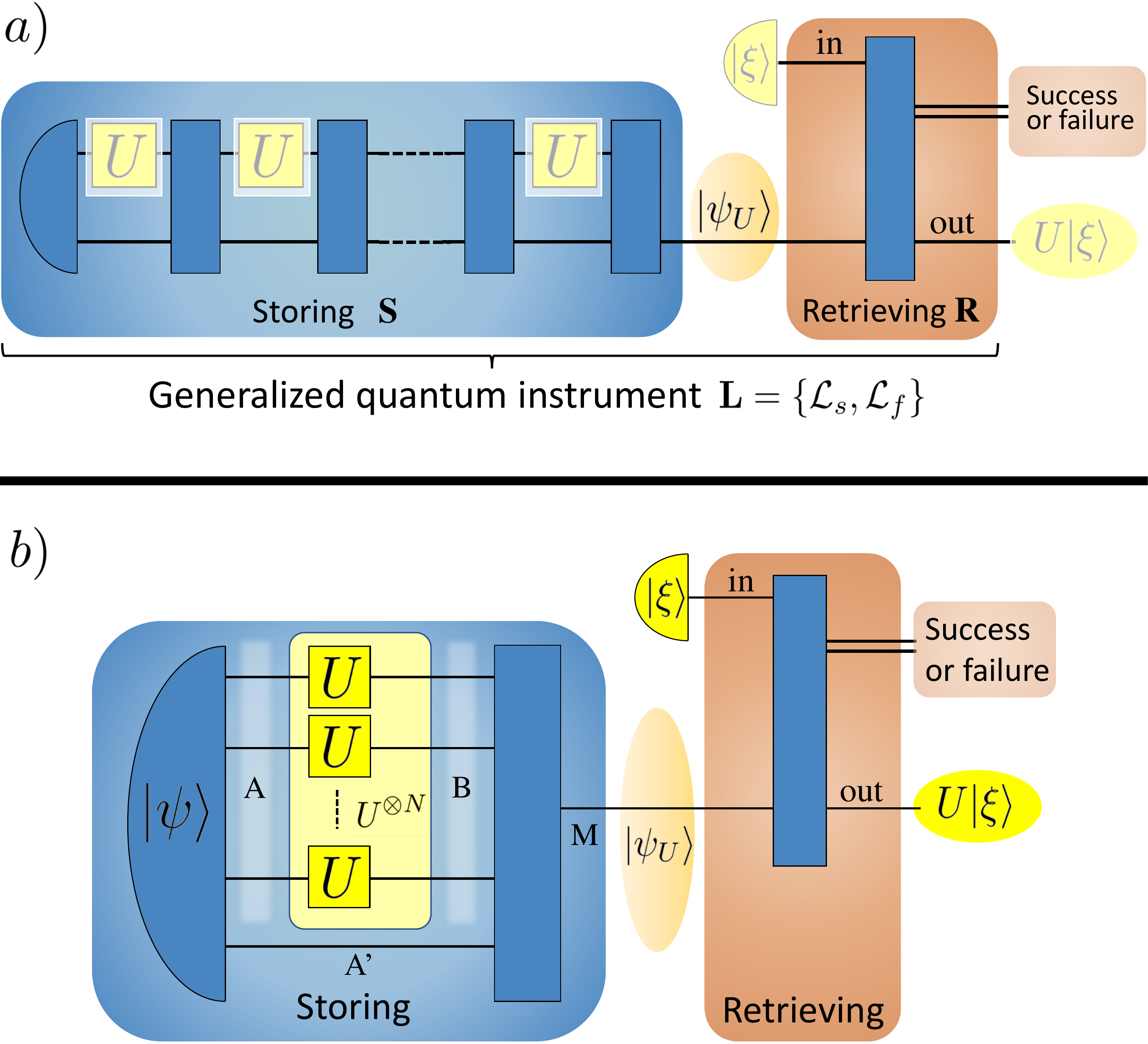}
    \caption{Illustration of $N\rightarrow 1$ PSAR
\emph{Top}: PSAR with the most general strategy. \emph{Bottom}:
      PSAR with parallel use of unitary channels.}
    \label{fig:2}
  \end{center}
\end{figure}

\emph{N-to-one probabilistic storage-and-retrieval.}  The general
PSAR strategy with $N$ uses of a channel in the storing phase involves
all combinations of their parallel, successive and adaptive processing
and corresponds to a quantum circuit with open slots, where the $N$
uses of a channel can be inserted. Such framework is described within the
theory of \emph{quantum networks} \cite{architecture,supermaps,comblong,actacomb} and any
quantum circuit with open slots is represented by a positive operator
(see \cite{supplement} for a short introduction).  The storing
network is
described by an operator ${S}$. It accepts $N$ channels
as its input and it outputs a memory state $\ket{\psi_U}\in\hilb{H}_M$
(see Fig.~\ref{fig:2}a). As in $1\to 1$ case the retrieving
phase is described by a two-valued instrument
$\mathbf{R}=\{\mathcal{R}_s,\mathcal{R}_f\}$. The overall action of
PSAR is a composition of $\mathcal{S}$ and $\mathbf{R}$ determining a
generalized quantum instrument
$\mathbf{L}=\{\mathcal{L}_s, \mathcal{L}_f \}$.  In the
Choi picture the input of PSAR corresponds to
$\Ket{U}\Bra{U}^{\otimes
  N}\in\mathcal{L}(\hilb{H}_A\otimes\hilb{H}_B)$ and
$L_s\in\mathcal{L}(\hilb{H}_A\otimes\hilb{H}_B\otimes\hilb{H}_{\rm
  out}\otimes\hilb{H}_{\rm in})$, where
$\hilb{H}_A=\hilb{H}_B=\hilb{H}^{\otimes N}$. The perfect retrieval
condition (similarly to Eq. (\ref{eq:case11composition})) is
\begin{align}
\label{eq:perfectlearcond}
\Bra{U^*}^{\otimes N} L_s \Ket{U^*}^{\otimes N}
=\lambda \KetBra{U}{U} \quad\quad \forall U \in SU(d),
\end{align}
where $\lambda$ gives the success probability.
Let us stress that the probability of success, i.e. the value of
$\lambda$, is required to be the same for all $U \in SU(d)$. Thanks to
this assumption we can without loss of generality apply the methods of
\cite{bisilearn} to conclude that the optimal storing phase is
\emph{parallel} as illustrated in Fig.~\ref{fig:2}b.  Consider the
decomposition
$U^{\otimes N} = \bigoplus_{j\in {\rm Irr}(U^{\otimes N})} U_j \otimes
I_{m_j}$ into irreducible representations (IRRs), where $U_j$
is a unitary operator on $\hilb{H}_j$ and $I_{m_j}$ denotes the identity
operator on the multiplicity space.
This corresponds to the following decomposition of the Hilbert space
$\hilb{H}_A := \bigoplus_{j\in {\rm Irr}(U^{\otimes N})} \hilb{H}_j
\otimes \hilb{H}_{m_j}$, and we set $d_j=\dim(\hilb{H}_j)$. The result
of \cite{bisilearn} implies that the memory state $\ket{\psi}$ can be
taken of the following form
\begin{align}
  \label{eq:optstate}
  \begin{aligned}
  \ket{\psi} :=  \bigoplus_{j} \sqrt{\frac{p_j}{d_j} } \Ket{I_j} \in
  {\mathcal{H}}_M \quad\quad
p_j \geq 0, \; \sum_j p_j =1\,,
  \end{aligned}
\end{align}
where $I_j$ denotes the identity operator on $\hilb{H}_j$ and
${\hilb{H}}_M := \bigoplus_{j\in {\rm Irr}(U^{\otimes N})} \hilb{H}_j\otimes \hilb{H}_j
\subseteq \hilb{H}_A\otimes \hilb{H}_{A'} $. The state $\ket{\psi}$
undergoes the action of the unitary channels and becomes
$\ket{\psi_U} :=\bigoplus_j \sqrt{\frac{p_j}{d_j} } \Ket{U_j}$.
Clearly, $\ket{\psi_U}\in\hilb{H}_M$ for any $U$.

Let us now focus on the retrieving quantum instrument
$\mathbf{R}$ from $\mathcal{L}(\hilb{H}_{\rm in}\otimes\hilb{H}_M)$
to $\mathcal{L}(\hilb{H}_{\rm out})$, where in/out labels the
system on which the retrieved channel is applied. The perfect
retrieval condition is again given by Eq. (\ref{eq:case11composition})
with $\ket{\psi^*_U} = \bigoplus_j \sqrt{\frac{p_j}{d_j} }\Ket{U^*_j}$.
As a consequence of Eq.~(\ref{eq:perfectlearcond}) the optimal
Choi operator $R_s$ can be chosen to satisfy
the commutation relation
\begin{align}
  \label{eq:commretriev}
  \left [ R_s,U'^* V' \otimes U_{\rm in}\otimes V^*_{\rm out}   \right ]=0,
\end{align}
where $U' := \bigoplus_j {U}_j \otimes I_{j}$, $V' := \bigoplus_j {I}_j \otimes V_{j}$.
Thanks to Eq.~\eqref{eq:commretriev}, $U'\ket{\psi}= \ket{\psi_U}$ and
$\ket{\psi^*_I}=\ket{\psi}$
the perfect retrieval condition becomes
\begin{align}
  \label{eq:simplifiedlambda}
  \begin{aligned}
    \bra{\psi} R_s \ket{\psi}= \lambda \KetBra{I}{I}
  \end{aligned}
\end{align}
and the success probability reads
$\lambda=\frac{1}{d^2}\Bra{I}\bra{\psi} R_s \ket{\psi} \Ket{I}$.
Let us now consider the decomposition
\begin{align}
  \label{eq:decompopartial}
  \begin{aligned}
    U^*_j \otimes U = \bigoplus_{J\in {\rm Irr}(U^*_j \otimes U)} U_J \otimes I_{m^{(j)}_{J}},
\end{aligned}
\end{align}
which induces the Hilbert space decomposition $\hilb{H}_j \otimes \hilb{H} = \bigoplus_{J\in {\rm Irr}(U^*_j \otimes U)} \hilb{H}_J \otimes \hilb{H}_{m^{(j)}_{J}}$.
Let us denote by $\mathsf{j}_{JK}$ the set of values of
$j$ such that $ U_J\otimes V_K$ is in the decomposition of
$U^*_j\otimes V_j \otimes U \otimes V^*$.
Using Eqs. (\ref{eq:commretriev}) and (\ref{eq:decompopartial})
we can assume \cite{supplement} that
 $R_s = \bigoplus_{J} I_J \otimes I_J \otimes s^{(J)}$,
where $s^{(J)} := \sum_{j,j' \in \mathsf{j}_{JJ}} s^{(J)}_{jj'}
\KetBra{I_{m_J^{(j)}}}{I_{m_J^{(j')}}}$.  Given this the left
hand side of Eq. (\ref{eq:simplifiedlambda}) reads
\begin{align}
  \label{eq:nuJ}
  \bra{\psi}    R_s \ket{\psi} =\sum_J  \lambda_J \KetBra{I}{I} + \nu_J
  \left(  I - \tfrac{1}{d} \KetBra{I}{I}\right)\,,
\end{align}
where $\nu_J$ are specified in \cite{supplement},
$\lambda_{J} = \frac{d_J}{d^2} \bra{\phi_J}s^{(J)}\ket{\phi_J}$ and
$\ket{\phi_J}=\bigoplus_{j\in \mathsf{j}_{JJ}} \sqrt{\frac{p_j}{d_j}}
\Ket{I_{m_J^{(j)}}}$. Since $R_s \geq 0 $, the perfect
learning condition of
Eq.~\eqref{eq:simplifiedlambda} holds only if $\nu_J=0$ for all $J$.
Then, the success probability is
$\lambda=\sum_J \lambda_J$.  The following result translates the
optimisation of $\lambda$ from an operator optimisation problem into a
linear program.
\begin{theorem}
  For optimal PSAR the success probability $\lambda$ is given by
  the following linear programming problem:
 \begin{align}
  \label{eq:optimizationpjmj}
    & \underset{\mu_J, p_j}{\mbox{\rm maximize}}
& & \lambda =\sum_{J\in\usbl} d_J^3 \mu_{J},
\\
& \mbox{\rm subject to}
&&0 \leq d_J \mu_J  \leq \frac{p_j}{d_j^2} \quad \forall j \in \mathsf{j}_{JJ}  \;\;\;\forall J\in \usbl  \nonumber  \\
 &&& p_j \geq 0 \quad\quad  \sum_{j\in {\rm Irr}(U^{\otimes N})}
 p_j=1\; , \nonumber
 \end{align}
 where $\usbl=\{J\in {\rm Irr}(U^{\otimes N}\otimes U^*) | d d_J =
 \sum_{j\in\mathsf{j}_{JJ}}d_j \}$.
\end{theorem}
\begin{proof}
  We will sketch only the key steps. The complete proof is in \cite{supplement}.
  First, one shows that  $J\notin \usbl$ implies $s^{(J)}=0$.
  Then,  (for any  $J\in \usbl$)  $\nu_J=0$ and $s^{(J)}\geq 0$ imply that
  $\sqrt{{p_j p_{j'}}} s^{(J)}_{jj'} =\mu_{J} \sqrt{d^3_j d^3_{j'}}$
  for some $\mu_{J}\geq 0 $.
Thus, 
$\lambda=\sum_{J\in \usbl} \sum_{j,j' \in \mathsf{j}_{JJ}} \frac{d_J \mu_{J}}{d^2} d_j d_{j'} =
  \sum_{J\in\usbl} d_J^3 \mu_{J}$.
The constraint that $\mathcal{R}_s$ is a quantum operation gives ${\rm tr}_{\rm out}[R_s]\leq I$.
Eq.~\eqref{eq:commretriev} implies $ \left[{\rm tr}_{\rm out}[R_s], U'V' \otimes U^*_{\rm in} \right]=0$ and ${\rm tr}_{\rm out}[R_s] = \bigoplus_{J}\bigoplus_{j\in \mathsf{j}_{JJ}} I_J \otimes I_j \,\frac{d_J}{d_j}\, s^{(J)}_{jj}$.
Thus, 
  $d_J \mu_J \frac{d_j^2}{p_j}\leq 1$ must hold for all $J$ and $j\in  \mathsf{j}_{JJ}$. Conditions on $p_j$ are from Eq. (\ref{eq:optstate}).
\end{proof}

\emph{Case study: $N\rightarrow 1$ PSAR for qubit channels.}
In case of qubit ($d=2$)  the decomposition of  $U^{\otimes N}$
into IRRs of $SU(2)$ reads
$U^{\otimes N}=\bigoplus_{j=(N \mod 2)/2}^{N/2} U_j \otimes I_{m_j}$,
where $m_j=\frac{2j+1}{N/2+j+1} \binom{N}{N/2+j}$ \cite{su2decomp} and
$U_j$ are the IRRs of spin $j$ with dimension $d_j=2j+1$.
For convenience we work with even $N$ (for odd $N$ see
\cite{supplement}), so $j=0,1,\ldots,N/2$. For
$SU(2)$ the complex conjugate representation
$U_j^*$ is equivalent to IRR $U_j$.
Thus, in Eq. (\ref{eq:decompopartial}) we get either $J=j+1/2$ or
$J=j-1/2$. Altogether, $J$ can have values $J\in \usbl =\{1/2, 
\ldots, (N-1)/2\}$ or $J=(N+1)/2\notin
\usbl$, because $\sum_{j\in\mathsf{j}_{JJ}}d_j=d_{J-1/2}+d_{J+1/2}=d
d_J$ and $d_{N/2}\neq 2 d_{(N+1)/2}$. The
constraints in Eq. (\ref{eq:optimizationpjmj}) imply for any
$j$ but $j=0,N/2$ the following two inequalities
\begin{align}
  \label{eq:ineq1}
  \mu_{j+1/2}\; d_j^2 d_{j+1/2} \leq p_j\,, \\
  \label{eq:ineq2}
  \mu_{j-1/2}\; d_j^2 d_{j-1/2} \leq p_j\,.
\end{align}
For $j=0,
N/2$ only one of them exists. Let us define $f_j\in
[0,1]$ for $j=0, \ldots
,\frac{N}{2}$ as
$f_j=\frac{1}{2}\frac{2j}{2j+1}\left(\frac{2j+2}{N}+1\right)$. Since
$f_0=0$ and
$f_{N/2}=1$ we can multiply Eq. (\ref{eq:ineq1}) by
$1-f_j$ and Eq. (\ref{eq:ineq2}) by
$f_j$, and take the sum for all
$j$.  A straightforward calculation gives the upper bound:
\begin{align}
\label{eq:mainub}
\frac{N+3}{N} \sum_{J=\frac12}^{\frac{N-1}{2}} d_J^3 \mu_J \leq 1
  \quad \Leftrightarrow \quad
\lambda \leq \frac{N}{N+3}.
\end{align}
Finally, by choosing $p_j=(2j+1)^2/L$,
$\mu_{j+1/2}=1/(L(2j+2))$ (where
$L=(N+1)(N+2)(N+3)/6$), one proves that that conditions in
Eq. (\ref{eq:optimizationpjmj}) are satisfied and the upper
bound (\ref{eq:mainub}) is achieved.
The knowledge of
$\mu_J$ and
$p_j$ completely specifies the state
$\ket{\psi}$ and the retrieving operation
$\mathcal{R}_s$ which can be explicitly expressed (see Fig. \ref{fig:2}b).
Let $\ket{j,j_z}\in \hilb{H}_j$ with
$j_z\in\{-j,\dots,j\}$ be an orthonormal basis of the spin $j$
IRR. By definition
$\Ket{I_j}=\sum_{j_z=-j}^j\ket{j,j_z}\otimes\ket{j,j_z}$. Consequently,
 from Eq.~\eqref{eq:optstate},
the dimension of the quantum memory is $\dim{\hilb{H}_M} =\sum_{j=0}^{N/2} d_j^2 = L$ and the optimal input state for storage is
$\ket{\psi}=  \bigoplus_{j=0}^{N/2} \sqrt{\frac{2j+1}{L}} \Ket{I_j}$.


\emph{Optimal PSAR for qudit unitary transformations.}
The optimization of $N\to 1$ PSAR of qudit channels
follows similar steps as for the qubit case and it exploits a
combinatorial identity (Proposition $3$ in~\cite{sedlak2018})
which was discovered and proved as a byproduct of this analysis.

\begin{theorem}
\label{thm1}
The optimal probability of success of $N\to 1$ probabilistic storage
and retrieval of a 
unitary channel $\chu(.)=U . U^\dagger$, $U\in SU(d)$ equals
$\lambda=N/(N-1+d^2)$.
The optimal state
for storage is
$
  \ket{\psi} :=  \bigoplus_{j }
  \sqrt{\frac{d_j}{L} } \Ket{I_j} \in  {\mathcal{H}}_M$
where
$ L := \sum_{j} d_j^2
$
and $ j\in {\rm
  Irr}(U^{\otimes N})$.
\end{theorem}
The proof is given in \cite{supplement}.  Clearly,
as $N$ goes to infinity $\lambda \sim 1 - \frac{d^2-1} {N}$, and
$\lambda\approx \frac12$ implies $N\approx d^2$. Reminding that a
$d$--dimensional unitary transformation has $d^2$ parameters, we see
that roughly one use per unknown parameter is needed for reliable storage and retrieval of the transformation.
Let us note that
the storage state in Theorem~\ref{thm1} is optimal also for the
estimation of a group transformation in the maximum likelihood
approach \cite{maxlikeest}. Further, it is worth to stress that
the optimal PSAR protocol is achieved by a coherent retrieval,
hence, the quantum memory is essential.
In contrast, optimal approximate SAR \cite{bisilearn} is equivalent to quantum estimation in the maximum fidelity approach and classical memory
is sufficient as an output of the storing phase. 
Use of the optimal storage state in the design of an approximate SAR leads to fidelity that scales as
$1-O(N^{-1})$, however, for the optimal approximate
SAR the fidelity scales as $1-O(N^{-2})$ \cite{bisilearn}.
This $O(N)$ difference is the price to pay for the
perfect retrieval in case of PSAR.


{\it Alignment of reference frames \cite{refframes}.} (ARF)
Let us note that the correction of alignment errors can be
modeled as a PSAR protocol in which $N$ uses of an unknown $\chu$ are
stored and the aim is to retrieve the inverse transformation
$\chu^{\dag}$.  For $SU(2)$, we can show that, given
$N$ uses of $\chu$,  the inverse transformation $\chu^{-1}$
can be perfectly retrieved with the same optimal probability of
success $\lambda$ (see Fig. \ref{fig:realignement} and \cite{supplement}).
It follows that the success probability of the probabilistic ARF
protocol \cite{refframes} achieves the optimal scaling $O(N^{-1})$
(see \cite{supplement}).


\begin{figure}
  \begin{center}
    \includegraphics[width=7cm]{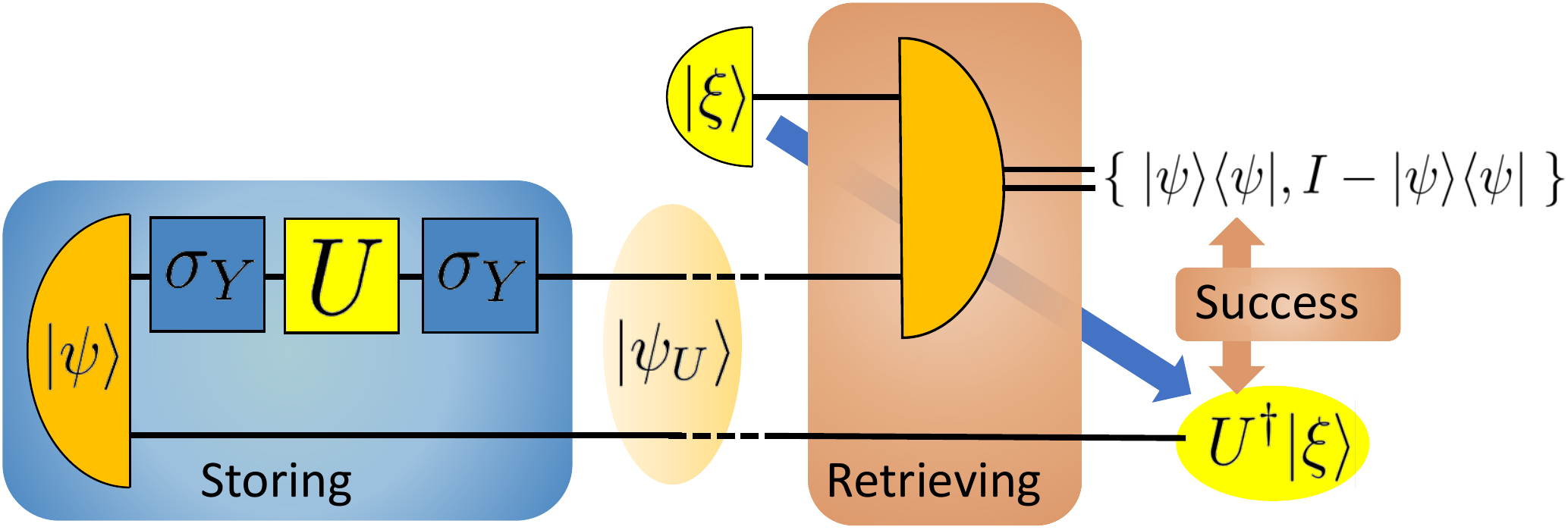}
    \caption{A modified optimal $1\rightarrow 1$ PSAR
      in which $\chu$ is stored and the inverse transformation
      $\chu^{\dag}$ is retrieved ($SU(2)$ case).
    The generalisation to the $N\rightarrow 1$ is straightforward. }
    \label{fig:realignement}
  \end{center}
\end{figure}

\emph{Probabilistic port-based teleportation.} (PPBT) As
the first step of PPBT \cite{portIshizaka1} Alice and Bob share
$N$ suitably entangled pairs of quantum systems.
Their goal is to teleport an unknown state
$\xi$ to Bob in a way that this state appears in one of his systems
(called ports \cite{portishizaka0,portstrelchuk2}). 
In order to achieve this goal (see also Fig.~\ref{fig:reltoportbased})
Alice performs a specific measurement
resulting in $n\in\{0,1,\dots,N\}$ ($0$ labels the failure of the protocol),
communicates this information to Bob who selects the system from
the $n$th port to accomplish the teleportation. If
Bob applies a channel $\chu$ on each of his ports (storing phase) and
Alice starts the teleportation (retrieving phase) of $\xi$ afterwards,
the $n$th port will output $\chu(\xi)$. Strictly speaking, we swap
$n$th port into a fixed quantum system and effectively we achieve
$N\rightarrow 1$ PSAR. Let us stress that while any PPBT protocol can
be turned into a PSAR protocol, the converse does not hold. In a
sense, PPBT scheme provides a structurally simple realization of an
optimal PSAR protocol.  Our results show that the optimal probability
of PPBT \cite{portStrelchuk1} coincides with the optimal success
probability of PSAR. However, the memory dimension $\dim\hilb{H}_M$
of the optimal PSAR is exponentially smaller (see the following
paragraph) in comparison with $2N$ qudits used in PPBT construction.
\begin{figure}
  \begin{center}
    \includegraphics[width=6cm]{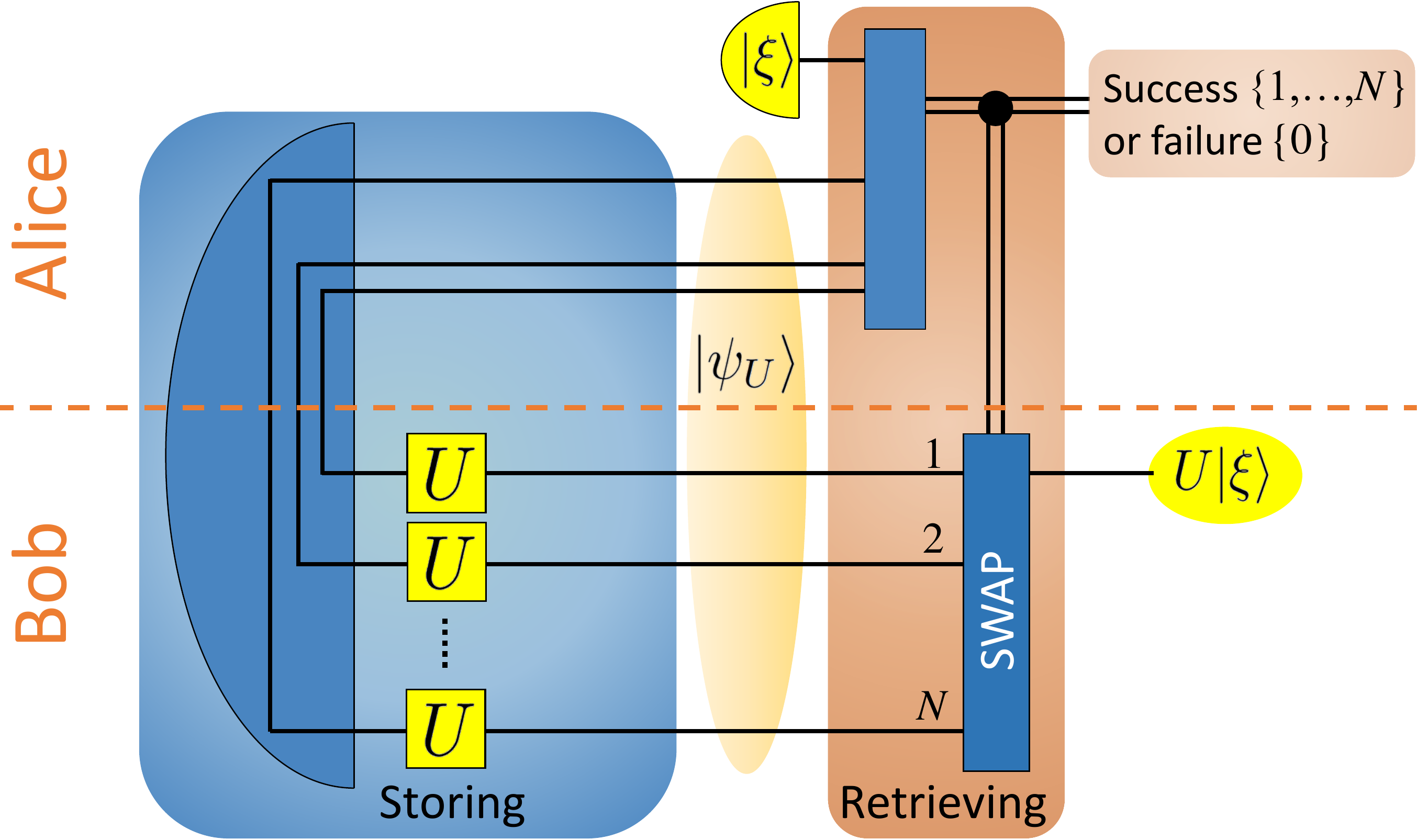}
    \caption{Use of port-based teleportation scheme for PSAR.}
    \label{fig:reltoportbased}
  \end{center}
\end{figure}

\emph{Implications for covariant probabilistic programmable
  processors.}
Up to now the best bound on the size of the program register for
universal covariant probabilistic processors was provided by family of
PPBT processors for which $\dim\hilb{H}_M \approx (d^{2(d^2-1)})^{1/f}$, where
$f=1-\lambda$ is the failure probability.
In contrast, the retrieving phase of optimal $N\rightarrow 1$ PSAR
defines a class of processors for which the program register size reads
$\dim \hilb{H}_M=\sum_{j\in {\rm Irr}(U^{\otimes N})} d_j^2 =
\binom{N+d^2-1}{N}$, where we used Schur's result \cite{schurthesis}.
In terms of the failure probability
it reads $\dim \hilb{H}_M \propto(1/f)^{(d^2-1)}$,
which is exponentially smaller (for fixed $d$ and $f\to 0$)
in comparison with PPBT-based processors. This result
can be viewed as a quantification of achievable tradeoffs imposed by
the no-programming theorem \cite{nielsen1} on universal covariant probabilistic
processors. Although PSAR provides only an upper bound on the size of
the program register, we conjecture that the lower bound will have the
same scaling. However, this question remains open.

\emph{Summary.}  We showed that optimal probabilistic
storage-and-retrieval of unknown unitary channels on $d$-dimensional
quantum systems can be designed with success probability
$\lambda=N/(N-1+d^2)$, where $N$ is the number of uses of the channel
in the storing phase.  This probability coincides with the success
probability for probabilistic port-based teleportation
\cite{portStrelchuk1}, and, for the $SU(2)$ case, with the
probability of success for probabilistic alignment of reference
frames. Optimal PPBT can be rephrased as an optimal
protocol for PSAR, but for the PSAR protocol designed here the
storing memory system is exponentially smaller and optimal in this
parameter. On the other hand, $N\to 1$ PPBT-based
PSAR implements all quantum channels (not only unitary ones), thus,
its performance is universal. The question of potential reduction of
memory system while keeping the universality for all channels remains
open.

\begin{acknowledgments}
MS and MZ
acknowledge the support by the QuantERA project HIPHOP (project ID 731473), projects QETWORK (APVV-14-0878), MAXAP (VEGA 2/0173/17), GRUPIK (MUNI/G/1211/2017) and the support of the Czech Grant Agency (GA\v CR) project no. GA16-22211S.
AB acknowledges the support of the John Templeton Foundation under the
project ID\# 60609 Causal Quantum Structures. The opinions expressed
in this publication are those of the authors and do not necessarily
reflect the views of the John Templeton Foundation.
\end{acknowledgments}

\newpage

\section{SUPPLEMENTAL MATERIAL}
This Supplemental Material provides a short introduction to theory of quantum networks, detailed proofs of Theorems 1,2 and more precise clarification of the relation of the presented work to the alignment of reference frames.

\section{Quantum networks and generalized instruments}
\label{sec:quant-netw-gener}

The mathematical formalization of the perfect learning of a unitary
channel can be easily given within the framework of \emph{quantum
  networks}. In this section we provide a small review of the subject and we
refer to the literature  \cite{architecture,comblong,actacomb} for a complete presentation.

We will start by introducing some notation.
If $\hilb{H}$ and $\hilb{K}$ are  finite-dimensional Hilbert spaces,
then we denote with  $\mathcal{L}(\hilb{H})$
the set of linear operator on $\hilb{H}$ and with $\mathcal{L}(\hilb{H},\hilb{K})$
the set of linear operator from $\hilb{H}$ to $\hilb{K}$.
We will use the
one-to-one correspondence between linear operators
$A \in \mathcal{L} (\hilb H,\hilb K)$
and vectors $\Ket{A} \in \hilb K \otimes \hilb H$ and
given by
\begin{align}\label{doubleket}
\Ket{A} = \sum_{m =1}^{{\rm dim} (\hilb K)} \sum_{n=1}^{{\rm dim} (\hilb H)}\<m |A|n\>  |m\> |n\>,
\end{align}
where $\{|m\> \}_{m=1}^{{\rm dim} (\hilb K)}$ and
$\{|n\> \}_{n=1}^{{\rm dim} (\hilb H)}$ are two fixed orthonormal
bases for $\hilb K$ and $\hilb H$, respectively.
For $A,B$ and $C$ operators on $\hilb{H}$
one can verify the identity
\begin{align}\label{doubleketidentity}
A \otimes B \Ket{C} = \Ket{ABC^T}
\end{align}
where $X^T$ denotes the transpose of
$X$ with respect to the orthonormal basis
$\ket{n}$.
A quantum operation $\mathcal{O}$ from
$\mathcal{L}(\hilb{H})$
to
$\mathcal{L}(\hilb{K})$
is a completely positive trace non increasing map
which can be represented by its Choi operator
$O\in\mathcal{L}(\hilb{K} \otimes \hilb{H})$.
The operator $O$
must satisfy
\begin{align}
  \label{eq:choiop}
  O \geq 0, \quad \Tr_{\hilb{K}}[O] \leq I_{\hilb{H}}
\end{align}
where $\Tr_{\hilb{K}}$ denotes the partial trace on $\hilb{K}$ and
$I_{\hilb{H}}$ the identity operator on $\hilb{H}$.  The two
  constraints in Eq. \eqref{eq:choiop} correspond to the complete
  positivity and trace non increasing of the quantum operation
  $\mathcal{O}$.
By making use of the notation in Eq.~\eqref{doubleket},
the Choi operator for a unitary channel $\mathcal{U}$
can be written as the rank one projector
$\Ket{U}\Bra{U}$.

The action of the quantum operation $\mathcal{O}$ on
  a quantum state $\rho \in \mathcal{L}(\hilb{H})$ can be described in
  terms of the Choi operator $O$ as follows
\begin{align}
  \label{eq:choiaction}
  \mathcal{O}(\rho) = \Tr_{\hilb{K}}[O (I_{\mathcal{K}}\otimes
  \rho^T)]=: O * \rho
\end{align}
where we introduce the \emph{link product} between the operators $O$
and $\rho$.  The composition of two quantum operations can be
represented in terms of their Choi operators too.
Let us consider two
quantum operations $\mathcal{O},\mathcal{O}'$
 with multipartite input and output, i.e.
$\mathcal{O}$ goes from
$\mathcal{L}(\hilb{H}_1\otimes \hilb{H}_2)$
to
$\mathcal{L}(\hilb{H}_3\otimes \hilb{K})$
and
$\mathcal{O}'$ goes from
$\mathcal{L}(\hilb{H}_4\otimes \hilb{K})$
to
$\mathcal{L}(\hilb{H}_5\otimes \hilb{H}_6)$.
We can connect the output of
$\mathcal{O}$ on $\mathcal{L}(\hilb{K})$
with the input
of
$\mathcal{O}'$ on $\mathcal{L}(\hilb{K})$
obtaining a new quantum operation from
$\mathcal{L}(\hilb{H}_1\otimes \hilb{H}_2 \otimes \hilb{H}_4)$
to
$\mathcal{L}(\hilb{H}_3\otimes \hilb{H}_5 \otimes \hilb{H}_6)$.
The Choi operator of the resulting quantum operation is given by the
link product of the two quantum operations, as follows:
\begin{align}
  \label{eq:linkoperation}
O'*O = \Tr_{\hilb{K}} [(O' \otimes I_{456}) ( I_{123} \otimes O^{T_\mathcal{K}} ) ]
\end{align}
where  $O^{T_\mathcal{K}}$ denotes the partial transposition of $O$
on the Hilbert space $\mathcal{K}$
and $I_{ijk}$ denotes the identity operator on
$\hilb{H}_i \otimes \hilb{H}_j \otimes \hilb{H}_k$.
We can interpret Eq.~\eqref{eq:choiaction} as an instance
of Eq. \eqref{eq:linkoperation}.

A quantum network $\mathcal{R}$ consists in a sequence of multipartite
quantum operations $ \{ \mathcal{O}_i , i=1,\dots N \} $ where some
output of a $ \mathcal{O}_i $ is connected to some input of the
following quantum operation $ \mathcal{O}_{i+1} $ as we illustrate in
the following diagram:
\begin{align}
  \begin{aligned}
\Qcircuit @C=0.5em @R=1em {
    \ustick{\scriptstyle{0}}&\multigate{1}{\mathcal{O}_1}&\ustick{\scriptstyle{1}}\qw&&\ustick{\scriptstyle{2}}&\multigate{1}{\mathcal{O}_2}&\ustick{\scriptstyle{3}}\qw&\pureghost{\dots}&\ustick{\scriptstyle{2N-2}}&\multigate{1}{\mathcal{O}_N}&\ustick{\scriptstyle{2N-1}}\qw\\
    &\pureghost{\mathcal{O}_1}&\qw&\qw&\qw&\ghost{\mathcal{O}_2}&\qw&\cdots&&\ghost{\mathcal{O}_N}&}
  \end{aligned},
\label{eq:qnetwork}
\end{align}
where the folating wires correspond to the input and output systems of the quantum network.
$\mathcal{R}$  is called a \emph{deterministic} quantum network if all the quantum
operations in Eq.~\eqref{eq:qnetwork} are trace preserving, and it is
called a \emph{probabilistic} quantum network otherwise.

A quantum network can be represented by a Choi operator (commonly
called \emph{quantum comb}) which is given by the link product of all
the component quantum operations.
The Choi operator $R$ of a deterministic quantum network $\mathcal{R}$
obeys the following constraints
\begin{align}\label{recnorm}
\Tr_{2k-1} [ R^{(k)}] = I_{2k-2} \otimes R^{(k-1)} \qquad k=1, \dots, N~
\end{align}
where, referring to the diagram in Eq.~\eqref{eq:qnetwork}, the Hilbert
space of the wire labelled by $j$ is $\hilb{H}_j$,  $R^{(N)}=R$, $R^{(0)} =1$,
$R^{(k)} \in \mathcal{L}(\hilb{H}_{{odd}_k} \otimes
\hilb{H}_{{even}_k})$
with $\hilb{H}_{{even}_k} = \bigotimes_{j=0}^{k-1} \hilb{H}_{2j}$ and
$\hilb{H}_{{odd}_k} = \bigotimes_{j=0}^{k-1}
\hilb{H}_{2j+1}$. $R^{(k)}$ is the
Choi operator of the reduced network $\mathcal R^{(k)}$ obtained by
discarding the last $N-k$ teeth.  The set of of positive operators satisfying Eq. (\ref{recnorm})
and the set of deterministic quantum networks are in one to one
correspondence. On the other hand, a given deterministic quantum network
$\mathcal R$ can be realized as a composition of quantum channels
in many different ways.
In the probabilistic case,
the Choi operator of
 a probabilistic quantum network $\mathcal{T}$,
must satisfy
\begin{align}\label{eq:recnormprob}
  0 \leq T \leq R
\end{align}
where $R$ is the Choi operator of a deterministic
quantum network.
A given probabilitic quantum network
$\mathcal T$ can be realised as a composition of quantum operations
in many different ways. In particular, any
 probabilitic quantum network
$\mathcal T$ can be realised by a composition of channels $\{ \mathcal{C}\}$ and a final
quantum operation $\mathcal{O}$ as follows:
\begin{align}
\mathcal{T}=
  \begin{aligned}
\Qcircuit @C=0.5em @R=1em {
    \ustick{\scriptstyle{0}}&\multigate{1}{\mathcal{C}_1}&\ustick{\scriptstyle{1}}\qw&&\ustick{\scriptstyle{2}}&\multigate{1}{\mathcal{C}_2}&\ustick{\scriptstyle{3}}\qw&\pureghost{\dots}&\ustick{\scriptstyle{2N-2}}&\multigate{1}{\mathcal{O}}&\ustick{\scriptstyle{2N-1}}\qw\\
    &\pureghost{\mathcal{C}_1}&\qw&\qw&\qw&\ghost{\mathcal{C}_2}&\qw&\cdots&&\ghost{\mathcal{O}}&}
  \end{aligned}.
\label{eq:probqnetwork}
\end{align}
A set of probabilistic quantum networks $\{\mathcal{R}_i\}$, with the
same input and output wires, is called a \emph{generalised quantum
  instrument} if the sum of their Choi operators
$\sum_i R_i =: R$ is the Choi operator of a deterministic quantum
networks.
As in the analogous case of quantum instruments, the index $i$
which labels the elements of a generalised quantum
  instrument represents the classical outcome which is available
after the quantum network has been provided with some input.
If the outcome $i$ is obtained then it means that the
probabilistic quantum network
$\mathcal{R}_i$ happened.
Any generalised quantum instrument can always be realised by a
 a composition of channels followed by a final
quantum intrument.
We notice that for any
probabilistic quantum network there exists a generalised quantum
instrument which it belongs to.

\section{Relevant sub-blocks of retrieving operation $R_s$}
As we stated in the main text,
Choi operator $R_s$ of the retrieving operation 
can be chosen to satisfy the commutation relation
\begin{align}
  \label{eq:commretriev}
  \left [ R_s,U'^* V' \otimes U_{\rm in}\otimes V^*_{\rm out}   \right ]=0,
\end{align}
where $                                                                                                                                                                                                                                    U' := \bigoplus_j {U}_j \otimes I_{j}$, $V' := \bigoplus_j {I}_j \otimes V_{j}$. We remind also the perfect retrieving condition
\begin{align}
\label{eq:simplifiedlambda}
\bra{\psi} R_s \ket{\psi}= \lambda \KetBra{I}{I}.
\end{align}
For convenience we placed here also the decomposition
\begin{align}
  \label{eq:decompopartial}
  \begin{aligned}
    U^*_j \otimes U = \bigoplus_{J\in {\rm Irr}(U^*_j \otimes U)} U_J \otimes I_{m^{(j)}_{J}},
\end{aligned}
\end{align}
which induces the Hilbert space decomposition
\begin{align}
\hilb{H}_j \otimes \hilb{H} = \bigoplus_{J\in {\rm Irr}(U^*_j \otimes U)} \hilb{H}_J \otimes \hilb{H}_{m^{(j)}_{J}}.
\end{align}
First, we notice that the multiplicity spaces
$\hilb{H}_{m^{(j)}_{J}} $ and  $\hilb{H}_{m^{(j)}_{K}} $ are one dimensional and therefore
$ I_{m^{(j)}_{J}}$ are rank one. From the Schur-Weyl duality, any irreducible
  representation $U_j$ of $SU(d)$ is in correspondence with a young diagram $Y_j$. The defining representation $U$ is represented by a
  single box $\Box$. One can verify that there cannot be two equivalent Young diagrams in the  decomposition $Y_j \times \Box = \sum_K Y_K$.
 For a more detailed treatment we refer to \cite{fulton2013}.
Then we have that
\begin{align}
  \label{eq:decompototal}
  \begin{aligned}
    U'V' \otimes U^*\otimes V^* =
    \bigoplus_{JK} U_J\otimes V_K \otimes I_{m_{JK}}
  \end{aligned}
\end{align}
induces decomposition
$\hilb{H}_{m_{JK}} = \bigoplus_{j\in \mathsf{j}_{JK}} \hilb{H}_{m_J^{(j)}}\otimes\hilb{H}_{m_K^{(j)}}$,
where $\mathsf{j}_{JK}$ denotes the set of values of
$j$ such that $ U_J\otimes V_K$ is in the decomposition of
$U^*_j\otimes V_j \otimes U \otimes V^*$.
Since $\dim(\hilb{H}_{m_J^{(j)}})=1$
we stress that $\BraKet{I_{m_J^{(j)}}}{I_{m_J^{(j')}}}=\delta_{j,j'}$
and $ \hilb{H}_{m_{JJ}} = \spn(\{\Ket{I_{m_J^{(j)}}}\}, j\in \mathsf{j}_{JJ})$.

From Eq.~\eqref{eq:decompototal} the commutation relation of
Eq.~\eqref{eq:commretriev}
becomes $[R_s ,  \bigoplus_{JK} U_J \otimes V_K \otimes I_{m_{JK}}]=0$,
which, thanks to the Schur's lemma, gives
\begin{align}
  \label{eq:decompor}
  \begin{aligned}
    R_s = \bigoplus_{J,K} I_J \otimes I_K \otimes s^{(JK)},
  \end{aligned}
\end{align}
where $s^{(JK)} \in \mathcal{L}(\hilb{H}_{m_{JK}})$, $s^{(JK)}\geq 0$.
Due to $\KetBra{I}{I}$ being a rank one operator and 
$R_s$ being the sum of the positive operators from Eq. \eqref{eq:decompor} we have that Eq. (\ref{eq:simplifiedlambda}) holds if and only if
\begin{align}
\label{eq:perfectforcomponents}
\bra{\psi}   I_J \otimes I_K \otimes s^{(JK)} \ket{\psi} =
\lambda_{JK} \KetBra{I}{I} \quad \forall J,K.
\end{align}
From the identity $I_j \otimes I = \bigoplus_{J\in {\rm Irr}(U^*_j \otimes U)}
I_J \otimes I_{m^{(j)}_{J}}$ (we remind that $ I_{m^{(j)}_{J}}$ has rank one),
we obtain
\begin{align}
  \label{eq:statedouble}
      \ket{\psi}\Ket{I} &= \bigoplus_j \bigoplus_{J\in {\rm Irr}(U^*_j \otimes U)}
\sqrt{\frac{p_j}{d_j}} \Ket{I_J}\Ket{I_{m_J^{(j)}}}
= \bigoplus_J
\Ket{I_J} \ket{\phi_J} \\
\ket{\phi_J}:&= \bigoplus_{j\in \mathsf{j}_{JJ}}
\sqrt{\frac{p_j}{d_j}} \Ket{I_{m_J^{(j)}}}.
\label{eq:relevantstate_app}
 \end{align}
Using Eqs.~\eqref{eq:decompor}, (\ref{eq:statedouble}) into $\lambda =\frac{1}{d^2}\Bra{I}\bra{\psi} R_s \ket{\psi} \Ket{I} $
we obtain
\begin{align}
\label{eq:lambdaJ_app}
  \lambda = \sum_J \lambda_{JJ} \quad \;\;
  \lambda_{JJ} = \frac{d_J}{d^2} \bra{\phi_J}s^{(JJ)}\ket{\phi_J}
 \end{align}
where the
$\lambda_{JK}$'s were defined in Eq.~\eqref{eq:perfectforcomponents}.
It is now easy to show that we can assume
\begin{align}
  \label{eq:diagonalR_app}
    R_s &= \bigoplus_{J} I_J \otimes I_J \otimes s^{(J)},
\end{align}
where $s^{(J)} := \sum_{j,j' \in \mathsf{j}_{JJ}} s^{(J)}_{jj'} \KetBra{I_{m_J^{(j)}}}{I_{m_J^{(j')}}}$.
Indeed, let $R_s'=\bigoplus_{JK} I_J \otimes I_K \otimes s'^{(JK)}$ be
the optimal quantum operation and let us define the operators
$R_s=\bigoplus_{J} I_J \otimes I_J \otimes s^{(J)}$
where $s^{(J)} =  s'^{(JJ)}$ and
 $R''_s=\bigoplus_{J\neq K} I_J \otimes I_K \otimes s^{(JK)}$.
Since both $R_s$ and $R_s''$ are positive and
$R_s+R_s'' = R_s'$, we have that $\Tr_D[R_s']\leq I$ implies
 $\Tr_D[R_s]\leq I$ i.e. $R_s$ is a quantum operation.
Finally, from Eq.~\eqref{eq:statedouble} we have that
 $\bra{\psi}R_s\ket{\psi} = \bra{\psi}R'_s\ket{\psi}$, thus proving
 that also $\{R_s,\ket{\psi} \}$ is an optimal solution
 of our optimization problem.

\section{Explicit form of the retrieved channel}
\label{sec:proof-oflemma1}
Due to commutation relation (\ref{eq:commretriev}) and the form of $R_s$ given by Eq. (\ref{eq:diagonalR_app})
the retrieved channel has the following Choi operator
\begin{align}
  \label{eq:nuJ_app}
  \bra{\psi}    R_s \ket{\psi} =\sum_J  \lambda_J \KetBra{I}{I} + \nu_J
  \left(  I - \tfrac{1}{d} \KetBra{I}{I}\right).
\end{align}
As we stated in the main text the perfect retrieving condition is satisfied if and only if $\nu_J=0$ for all $J$. This happens because positive-semidefiniteness
of $R_s^{(J)}:= I_J \otimes I_J \otimes s^{(J)}$ implies $\nu_J\geq 0$ and the requirement $\sum_J \nu_J=0$ implies that all the terms must vanish.
Let us study separately every operator $R^{(J)}_s$, which by definition satisfies the commutation relation of Eq.~\eqref{eq:commretriev}.
We have that
\begin{align}
\nonumber
  &\bra{\psi} R_s^{(J)} \ket{\psi} = \\
\nonumber
  &=\bra{\psi}  (   U'U'^* \otimes U^*\otimes U ) R_s^{(J)} (U'U'^*
  \otimes U^*\otimes U)^\dag \ket{\psi}\\
 &=( U^*\otimes U  ) \bra{\psi}    R_s^{(J)} \ket{\psi}
(U^*\otimes U)^\dag \;\; \forall U \label{eq:commutationforperfect}
  \end{align}
Thanks to the Schur's lemma Eq.~\eqref{eq:commutationforperfect}
gives
\begin{align}
  \label{eq:nuJ_app}
  \bra{\psi}    R_s^{(J)} \ket{\psi} = \lambda_J \KetBra{I}{I} + \nu_J
  \left(  I - \tfrac{1}{d} \KetBra{I}{I}\right).
\end{align}
By taking the trace of Eq.~\eqref{eq:nuJ_app} we have
\begin{align}
  \label{eq:computenuJ}
  \begin{aligned}
&  \Tr[\bra{\psi}    R_s^{(J)} \ket{\psi} ]=
\bra{\psi} \Tr_{out\;in}[    R_s^{(J)} ] \ket{\psi}=\\
&\bra{\psi} \bigoplus_{j\in \mathsf{j}_{JJ}} I_j\otimes I_j q^{(J)}_j\ket{\psi}=\\
&\sum_{j\in \mathsf{j}_{JJ}} p_j  q^{(J)}_j = \lambda_J d + \nu_J(d^2-1)
  \end{aligned}\\
q^{(J)}_j := \frac{d^2_J}{d^2_j} s^{(J)}_{jj}. \nonumber
\end{align}
If we insert Eq.~\eqref{eq:relevantstate_app} into Eq.~\eqref{eq:lambdaJ_app}
we have
\begin{align}
  \label{eq:lambdaJexplicit}
\lambda_J= \frac{d_J}{d^2} \sum_{j,j'\in\mathsf{j}_{JJ}}
    \sqrt{\frac{p_jp_{j'}}{d_jd_{j'}}}
s^{(J)}_{jj'}.
\end{align}
 From Eq.~\eqref{eq:computenuJ} and Eq.~\eqref{eq:lambdaJexplicit}
we have
\begin{align}
  \label{eq:nuJequal0}
  \begin{aligned}
  &\nu_J = 0 \iff\\
&\sum_{j,j'\in\mathsf{j}_{JJ}}
\!\!\! \delta_{j,j'}\,
d\, d_J\frac{p_j}{d^2_j} s^{(J)}_{jj}
-\sqrt{\frac{p_jp_{j'}}{d_jd_{j'}}}
s^{(J)}_{jj'}       =0,
  \end{aligned}
\end{align}
which is the most explicit form of the perfect retrieving condition that constraints the relation between the state $\ket{\psi}$ parametrized by probabilities $p_j$ and the structure of the retrieving operation parameterized by $s^{(J)}_{jk}$.


\section{$N\rightarrow 1$ PSAR as a linear programming problem}
In this section we provide complete proof of Theorem $1$ from the main text.
First, we prove the following technical lemma, which will be needed.

\begin{lemma}
\label{lmm:specialX}
Suppose a 
matrix $X=\sum_{j,j'} X_{jj'}\ket{j}\bra{j'}\geq 0$ obeys
$\sum_{j,j'} X_{jj'} = \sum_{j} \frac{1}{c_j} X_{jj}$, where $c_j>0$ and $\sum_j c_j=1$. This implies
$X \propto \ket{\chi}\bra{\chi}$, where $\ket{\chi}=\sum_j c_j \ket{j}$.
\end{lemma}

\begin{proof}
Let us define
\begin{align}
  \label{eq:9}
  \ket{v} &:= \sum_{j}\ket{j} \\
  \ket{\rho} &:= \sum_j \sqrt{X_{jj}} \ket{j} \\
  A &:= \sum_{j}\frac{1}{c_j} X_{jj}\ketbra{j}{j}
\end{align}
The condition $\sum_{j,j'} X_{jj'} = \sum_{j} \frac{1}{c_j} X_{jj}$ can be written as
\begin{align}
  \label{eq:10}
  \bra{v} A -X \ket{v} = 0.
\end{align}
Matrix $H_{ij}$ is positive semidefinite if and only if $H_{ii}\geq 0 \;\;\forall i$ and $|H_{ij}|\leq \sqrt{H_{ii}H_{jj}}$ $\forall i\neq j$. Using this criterion and $\sum_j c_j=1$ one can easily show that both $A-X$  and $A -\ketbra{\rho}{\rho} $ are  positive semidefinite matrices.
Moreover, using $\Re(X_{jj'})\leq |X_{jj'}|\leq \sqrt{X_{jj}X_{j'j'}}$ one can easily prove the inequality
\begin{align}
  \label{eq:12}
\bra{v}( A -X )\ket{v} \geq   \bra{v} (A -\ketbra{\rho}{\rho}) \ket{v},
\end{align}
which also gives
\begin{align}
  \label{eq:13}
  \bra{v} (A -X) \ket{v}  = 0  \implies   \bra{v}(A -\ketbra{\rho}{\rho})
  \ket{v} =0,
\end{align}
due to $A -\ketbra{\rho}{\rho}\geq 0$.
Moreover, let us rewrite expression $\bra{v} A \ket{v}$ as
\begin{align}
  \label{eq:14}
  \bra{v} A \ket{v}  = \bra{\rho} B  \ket{\rho} \\
  B :=  \sum_{j}\frac{1}{c_j} \ketbra{j}{j}
\end{align}
As a consequence we have
\begin{align}
  \label{eq:15}
   \bra{v}(A -\ketbra{\rho}{\rho})
  \ket{v} = \bra{\rho} ( B - \ketbra{v}{v} )\ket{\rho}.
\end{align}
Thanks to $c_j>0 \;\;\forall j$ we have that $B - \ketbra{v}{v}$ is a positive matrix, which has either trivial or one dimensional kernel. This together with Eqs. (\ref{eq:13}),(\ref{eq:15}) allows us to write a necessary condition for matrix $X$
\begin{align}
  \label{eq:16}
 \bra{\rho} ( B - \ketbra{v}{v} )\ket{\rho} =0 \\
  \implies
  B  \ket{\rho} - \ket{v}   \braket{v}{\rho} =0
\end{align}
Explicitly solving the above equation we get the only possible solution
\begin{align}
 \label{eq:17}
  \frac{1}{c_j}\sqrt{X_{jj}}  =  \frac{1}{c_{j'}}\sqrt{X_{j'j'}}
  \implies
  X_{jj} = \mu c_j^2,
\end{align}
which is unique up to a constant $\mu$, as we expected due to the rank one deficiency of $B - \ketbra{v}{v}$.
Once the diagonal elements $X_{jj}$ respect Eq. (\ref{eq:17}) we have  $\bra{v}(A -\ketbra{\rho}{\rho})\ket{v}=0$, but to fulfill LHS of Eq. (\ref{eq:13}) we need also the saturation of the bound (\ref{eq:12}). This happens if and only if
\begin{align}
 \label{eq:18}
  X_{jj'} = \sqrt{X_{jj}}\sqrt{X_{j'j'}},
\end{align}
which together with Eq. (\ref{eq:17}) proves the claim of the lemma.
\end{proof}

Let us restate Theorem 1 from the main text.

\begin{theorem}
\label{thm:lp}
  For optimal PSAR the success probability $\lambda$ is given by
  the following linear programming problem:
 \begin{align}
  \label{eq:optimizationpjmj}
    & \underset{\mu_J, p_j}{\mbox{\rm maximize}}
& & \lambda =\sum_{J\in\usbl} d_J^3 \mu_{J},
\\
& \mbox{\rm subject to}
&&0 \leq d_J \mu_J  \leq \frac{p_j}{d_j^2} \quad \forall j \in \mathsf{j}_{JJ}  \;\;\;\forall J\in \usbl  \nonumber \\
&&& p_j \geq 0 \quad \sum_{j\in {\rm Irr}(U^{\otimes N})} p_j=1\; , \nonumber
 \end{align}
 where $\usbl=\{J\in {\rm Irr}(U^{\otimes N}\otimes U^*) | d d_J = \sum_{j\in\mathsf{j}_{JJ}}d_j \}.$
\end{theorem}

\begin{proof}
We first need to examine relations between IRR's that appear in the decomposition of $U^{\otimes N}$ and
those that appear in $\bigoplus_{j\in {\rm Irr}(U^{\otimes N})} U^{(j)}\otimes U^*$.

We remind that from the Schur-Weyl duality, any irreducible representation $U_j$ of $SU(d)$ is in correspondence with a young diagram $Y_j$. The defining representation $U$ is represented by a single box $\Box$ and IRR defined via $U^*$ is represented by a column of $d-1$ boxes.

Decomposition of $U^{\otimes N}$ into IRRs can be obtained by collecting the decompositions of the tensor products $U_k \otimes U$ of all Young diagrams $k$ appearing with multiplicity $m_k$ in the decomposition of $U^{\otimes N-1}$ and putting together equivalent IRRs (those with the same Young diagram).
This can be mathematically stated as follows. Let ${\rm Irr}(U^{\otimes N})$ denote the set of Young diagrams that appear in the decomposition of $U^{\otimes N}$ into IRRs of $SU(d)$.
We have that $K\in {\rm Irr}(U^{\otimes N})$ if and only if $\exists k\in {\rm Irr}(U^{\otimes N-1})$ such that $K\in {\rm Irr}(U_k\otimes U)$ and $m_K=\sum_{k\in \mathsf{k}_K} m_k$, where $\mathsf{k}_{K}$ denotes the set of values of $k$ such that $U_K$ is in the decomposition of $U_k\otimes U$.
On the other hand, thanks to Schur-Weyl duality the multiplicity
$m_K=\DP_K$ ($m_k=\DP_k$)
is given by the dimension of the  IRRs of the symmetric group $S(N)$ ($S(N-1)$) with the Young diagram $K$ ($k$), respectively.
Hence, we obtained a known identity \cite{fulton2013}
\begin{align}
\label{eq:sumpg1}
\DP_K=\sum_{k\in \mathsf{k}_K} \DP_k,
\end{align}
where $\mathsf{k}_K$ can be equivalently specified as those Young diagrams $k$, which by addition of a single box become $K$.

Next, we consider decomposition
of 
$U_j \otimes U^*$ (or more conveniently $U^* \otimes U_j$), where
$j\in {\rm Irr}(U^{\otimes N})$.  We denote Young diagram $Y_j$ with
$r$ rows and $n_i$ boxes in the $i$-th row as $(n_1,n_2,\ldots, n_r)$.
A valid Young diagram of $SU(d)$ IRR has $r\leq d$, $n_{r}>0$ and
$n_i\geq n_{i+1}\;\;\forall i$ (we set $n_{r+1}=0$).  Rows $i$ in
which $n_i> n_{i+1}$ we call corners of $(n_1,n_2,\ldots, n_r)$ and we
denote the number of corners by $s$ and we write $i\in \corn_j$.
Suppose $Y_j \leftrightarrow (n_1,n_2,\ldots, n_r)$ has $r\leq
d-1$. Then the decomposition of $U^* \otimes U_j$ contains $s+1$ Young
diagrams each with multiplicity one. One of them is given as Young
diagram $(n_1+1,n_2+1,\ldots,n_r+1,1,\ldots, 1)$ with $d-1$ rows,
which we denote $Y_{|j}$ and for each $i\in\corn_j$ we have Young
diagram
$Y_{j\backslash \lrcorner i}\leftrightarrow (n_1,\ldots,n_i-1,\ldots,
n_r)$. The above statement follows from the Littlewood-Richardson
rules \cite{fulton2013} if one realizes, that either one of the corner boxes completes
the first column into $d$ boxes (the remaining boxes can be only
attached to the right in the original order) or the whole Young
diagram is attached from the right to the column of $d-1$ boxes.  If
$Y_j \leftrightarrow (n_1,n_2,\ldots, n_r)$ has $r=d$ the situation is
the same except for the diagram $Y_{|j}$ not appearing in the
decomposition, because it would not be a valid Young diagram.  Let us
note that Young diagram $Y_{|j}$ can emerge in our setting only from
diagram $Y_l$, where $l=j$. We can also easily verify that
$d d_{|j}\neq d_j$, which can be seen from the formula for the
dimension of $SU(d)$ IRRs \cite{Stanley1999} by calculating the
fraction $d_{|j}/d_j$ for a general $j$.  Therefore, we conclude that
for $J=|j$ Eq. (\ref{eq:nuJequal0}) can be satisfied only if
$s^{(|j)}_{jj}=0$, which in turn thanks to Eq. (\ref{eq:lambdaJ_app})
implies $\lambda_{|j}=0$.  Thus, Young diagrams
$\overline{\usbl}=\{Y_{|j},j\in {\rm Irr}(U^{\otimes N})\}$ correspond
to those $J$ that do not belong to the set $\usbl$ defined in the
theorem.

On the other hand, consistently with the notation for
$\overline{\usbl}$, we define
$\usbl=\{Y_{j\backslash \lrcorner i}, j\in {\rm Irr}(U^{\otimes N}),
i\in\corn_j\}$. Let us remind that ${\rm Irr}(U^{\otimes N})$ is
exactly constituted by all Young diagrams consisting of $N$ boxes and
having at most $d$ rows.  This implies
$\usbl={\rm Irr}(U^{\otimes N-1})$, because by removing in any
possible way a single box from Young diagrams in
${\rm Irr}(U^{\otimes N})$ we get all possible Young diagrams in
${\rm Irr}(U^{\otimes N-1})$. More operationally, for any Young
diagram $J\in\usbl$ we can add a box to the first row and get some
element $j\in{\rm Irr}(U^{\otimes N})$, which can be reversed to prove
the claim.

Moreover, for every subset $\usbl_j=\{Y_{j\backslash \lrcorner i}, i\in\corn_j\}$ of $\usbl$ we have that
\begin{align}
\label{eq:sumpg2}
\DP_j=\sum_{J\in \usbl_j} \DP_J,
\end{align}
which is just a reformulation of Eq. (\ref{eq:sumpg1}), because Young diagrams $J\in \usbl_j$ have $N-1$ boxes and an addition of a single box changes them to Young diagram $j$ consisting of $N$ boxes.

Let us pick any element $J\in \usbl$. Let us now specify all the Young diagrams $Y_j, \; j\in {\rm Irr}(U^{\otimes N})$, which contain $J$ in the decomposition of $U_j \otimes U^*$. We denote such set $\mathsf{j}_J$ and it coincides with $\mathsf{j}_{JJ}$ defined below Eq. (\ref{eq:decompototal}). These are such Young diagrams $j$ in which by removing one corner box we get $Y_J$. This is the same as saying that $\mathsf{j}_J$ is the set of Young diagrams of $SU(d)$ group that can be obtained from $J$ by addition of a single box,
because ${\rm Irr}(U^{\otimes N})$ 
contains all possibly emerging Young diagrams. 
This implies that $d d_J=\sum_{j\in \mathsf{j}_J} d_j$, because this corresponds to the decomposition of an operator $U_J\otimes U$, which acts on $d d_J$ dimensional space.
Thus, we proved that the set $\usbl$ can be equivalently defined as
\begin{align}
  \label{eq:crucialidentity}
  \usbl&=\{J\in {\rm Irr}(U^{\otimes N}\otimes U^*) | d d_J = \sum_{j\in\mathsf{j}_{JJ}}d_j \} \nonumber\\
  &=\{Y_{j\backslash \lrcorner i}, j\in {\rm Irr}(U^{\otimes N}), i\in\corn_j\} \nonumber\\
  &={\rm Irr}(U^{\otimes N-1})
\end{align}
Furthermore, we showed that for $J\notin \usbl$ $s^{J}=0$ and consequently $\lambda_J=0$.

In order to proceed we apply Lemma \ref{lmm:specialX} for every $J\in \usbl$. Expression $\sqrt{\frac{p_jp_{j'}}{d_jd_{j'}}} s^{(J)}_{jj'}$ plays the role of $X_{jj'}$, $c_j=\frac{d_j}{d d_J}$ and the remaining assumption is guaranteed by Eq. (\ref{eq:nuJequal0}).
As a consequence, we get that
the condition \eqref{eq:simplifiedlambda}
of perfect retrieving and $s^{(J)}\geq 0$ is equivalent to
\begin{align}
 \label{eq:rjjform}
   s^{(J)}_{jj'} =\mu_{J} \sqrt{\frac{d^3_j d^3_{j'}}{p_j p_{j'}}} \quad \quad \mu_{J}\geq 0 \quad \; \forall J\in \usbl
\end{align}
Thus, fulfillment of Eq. (\ref{eq:rjjform}) guarantees the perfect retrieving of unitary transformations and we can rewrite the probability of success as
\begin{align}
  \label{eq:19}
  \lambda = 
  \sum_{J\in \usbl} \sum_{j,j' \in \mathsf{j}_{JJ} } \frac{d_J}{d^2}\mu_{J} d_j d_{j'}  =
  \sum_{J\in\usbl} d_J^3 \mu_{J},
\end{align}
where we used Eqs. (\ref{eq:relevantstate_app}),(\ref{eq:lambdaJ_app}) and the defining property of the set $\usbl$.

The constraint that $\mathcal{R}_s$ is a quantum operation
translates into its Choi operator as
$\Tr_D[R_s]\leq I$. Since $R_s$ satisfies Eq.~\eqref{eq:commretriev},
we obtain $ \left[  \Tr_D[R_s],   U'V' \otimes U^*_C \right]=0 $, which implies
$\Tr_D[R_s] = \bigoplus_{J}\bigoplus_{j\in \mathsf{j}_{JJ}} I_J \otimes
I_j \,\frac{d_J}{d_j}\, s^{(J)}_{jj}$.
This implies
\begin{align}
  \label{eq:subnormalization}
  \begin{aligned}
    \Tr_D[R_s]\leq I \Leftrightarrow
\frac{d_J}{d_j} s^{(J)}_{jj} \leq 1 \; \; \forall J,\; \forall j\in  \mathsf{j}_{JJ}.
  \end{aligned}
\end{align}
Let us express 
the above condition 
via coefficients $\mu_J$ using Eq. (\ref{eq:rjjform})

\begin{align}
  \label{eq:20}
  \mu_J d_j^2 \leq \frac{p_j}{d_J} \quad \quad \; \; \forall J,\; \forall j\in  \mathsf{j}_{JJ}.
\end{align}
Let us remind the definition of state $\ket{\psi}$ from the main text.
\begin{align}
  \label{eq:optstate}
  \begin{aligned}
  \ket{\psi} :=  \bigoplus_j \sqrt{\frac{p_j}{d_j} } \Ket{I_j} \in
  \tilde{\mathcal{H}}\quad\quad
p_j \geq 0, \; \sum_j p_j =1\,,
  \end{aligned}
\end{align}
Collecting Eqs.(\ref{eq:19}), (\ref{eq:rjjform}),(\ref{eq:20}) and (\ref{eq:optstate}) we see that the optimization of probabilistic storage and retrieval is reduced to a linear program stated in the Theorem \ref{thm:lp}.
\end{proof}

\section{$N\rightarrow 1$ PSAR for qubit channels - The case of odd $N$}
\label{app:oddn}
All the steps are completely analogical to the derivation valid for even $N$ presented in the main text.
The main difference is that the IRR's with minimum and maximum spin ($J=0$ and $J=\frac{N+1}{2}$) have only multiplicity one.
For odd $N$ (identically as for even $N$) the investigation of the conditions of perfect learning reveals that $s^{\frac{N+1}{2}}$ has to be zero.
On the other hand, $J=0$ can be involved in the perfect storing and retrieving. Other expressions remain identical, but now $J$ is an integer.
In particular, we choose $f_J$ according to the same formula as in the main text 
\begin{align}
\label{eq:coeffj}
f_j=\frac{1}{2}\frac{2j}{2j+1}\left(\frac{2j+2}{N}+1\right)
\end{align}
and the whole proof goes on analogically to the case of even $N$.

\section{$N\rightarrow 1$ PSAR for qudit channels - The general case of $SU(d)$}

The goal of this section is to prove Theorem \ref{thm:Main} from the main text.
\begin{theorem}
\label{thm:Main}
The optimal probability of success of $N\to 1$ probabilistic storage
and retrieval of a 
unitary channel $\chu(.)=U . U^\dagger$, $U\in SU(d)$ equals
$\lambda=N/(N-1+d^2)$.
The optimal state
for storage is
$
  \ket{\psi} :=  \bigoplus_{j }
  \sqrt{\frac{d_j}{L} } \Ket{I_j} \in  {\mathcal{H}}_M$
where
$ L := \sum_{j} d_j^2
$
and $ j\in {\rm
  Irr}(U^{\otimes N})$.
\end{theorem}

\begin{proof}
The idea of the proof is analogical to the case of qubit unitary transformations. However, in the qudit case
the relations between IRRs are more complicated and we will need some of the facts derived in the proof of Theorem \ref{thm:lp} and a new combinatorial identities, which were derived in \cite{sedlak2018} by some of us.

Let us define positive function 
\begin{align}
\label{eq:deffjJ}
f(j,J)=\frac{\DP_J}{\DP_j},
\end{align}
for all $j \in {\rm Irr}(U^{\otimes N})$, $J\in \usbl_j$ or equivalently for all $J\in {\rm Irr}(U^{\otimes N-1})$, $j\in \mathsf{j}_{JJ}=\mathsf{j}_{J}$.
Let us note that thanks to Eq. (\ref{eq:sumpg2}) we have
\begin{align}
\label{eq:sumJfjJ}
\sum_{J\in \usbl_j} f(j,J)=1 \quad \forall j \in {\rm Irr}(U^{\otimes N}).
\end{align}
For the proof of the main theorem we need a new theorem from combinatorics \cite{sedlak2018} and a technical lemma.


\begin{theorem}
\label{thm:combthm1}
For any Young diagram $J$ consisting of $N-1$ boxes
it holds that
\begin{align}
  \label{eq:combthm1}
  \sum_{j \in \jJall} (C_j-R_j)^2 \DP_j = N(N-1)\; \DP_J,
\end{align}
where the sum runs through all Young diagrams $j$ that can be obtained from $J$ by addition of a single box,
\begin{itemize}
  \item $\DP_J, \DP_j$ are dimensions of IRRs of the symmetric group $S(N-1)$, $S(N)$, respectively
  \item $C_j$ is the number of the column of the added box,
  \item $R_j$ is the number of the row of the added box that leads from diagram $J$ to the diagram $j$.
\end{itemize}
\end{theorem}

\begin{lemma}
\label{lmm:identity42groups}
For any Young diagram $J\in {\rm Irr}(U^{\otimes N-1})$ the following
identity for dimensions $d_J, d_j$ of IRRs of $SU(d)$ group and for
dimensions of $\DP_J, \DP_j$ of the symmetric group, holds
\begin{align}
\label{eq:lemma2}
\sum_{j\in \mathsf{j}_{JJ}}  \frac{d_j^2}{\DP_j} =\frac{N-1 + d^2}{N}\;\;\frac{d_J^2}{\DP_J} \quad \quad \forall J\in {\rm Irr}(U^{\otimes N-1}),
\end{align}
where $\mathsf{j}_{JJ}=\{j\in {\rm Irr}(U^{\otimes N})\; | \; J\in{\rm Irr}(U_j \otimes U^*) \}$.
\end{lemma}

\begin{proof}
Let us remind expressions for the dimensions of IRRs that are involved (for detailed explanation see \cite{fulton2013}):
\begin{align}
\label{eq:defdim}
d_j =\frac{l_j}{h_j} \quad d_J =\frac{l_J}{h_J} \quad \DP_j=\frac{N!}{h_j} \quad \DP_J=\frac{(N-1)!}{h_J},
\end{align}
where $h_j, h_J$ denote the hook lengths factors and
$l_j$ is $\prod_{box \in Y_j} (d - R_i +C_i)$ (here $R_i$ ($C_i$) is the row (column) of the current box from the Young diagram $j$).
Using Eq. (\ref{eq:defdim}) we can write Eq. (\ref{eq:lemma2}) as
\begin{align}
  \sum_{j\in \mathsf{j}_{JJ}} \frac{{l_j}^2}{{l_J}^2}  \frac{{h_J}}{{h_j}} = d^2 + N - 1
  \label{eq:crucial}
\end{align}
Thus, proving Lemma \ref{lmm:identity42groups} is equivalent to proving that Eq. (\ref{eq:crucial}) holds.
We start by direct evaluation of the left hand side. We obtain:
\begin{align}
\label{eq:defleft1}
  \sum_{j\in \mathsf{j}_{JJ}} \frac{{l_j}^2}{{l_J}^2}  \frac{{h_J}}{{h_j}}
  &= \sum_{j\in \mathsf{j}_{JJ}} (d-R_j+C_j)^2 \frac{{h_J}}{{h_j}}
\end{align}
where $R_j$ is the row number and
$C_j$ the column number of the additional box in Young diagram $Y_j$ with respect to $Y_J$.
At this point it is useful to realize that for Young diagrams $J\in {\rm Irr}(U^{\otimes N-1})$ with $d$-rows, there is a difference between the set $\mathsf{j}_{JJ}=\mathsf{j}_{J}$ and the set $\jJall$ of all Young diagrams that can be obtained from $J$ by addition of a single box. The difference is exactly one Young diagram, which is obtained from $J$ by adding the box into the $d+1$-th row, in the first column. Luckily, the bracket $(d-R_j+C_j)$ for this diagram evaluates to zero ($d-(d+1)+1=0$), so we can sum also through this term in Eq. (\ref{eq:defleft1}) without changing its value. This is useful especially for $d < N$, because later on we want to apply Theorem \ref{thm:combthm1}, where the summation runs through the set $\jJall$. Thus,
left hand side of Eq. (\ref{eq:crucial}) can be equivalently rewritten as
\begin{align}
\label{eq:defFGH}
  \sum_{j\in \mathsf{j}_{JJ}} \frac{{l_j}^2}{{l_J}^2}  \frac{{h_J}}{{h_j}}
  = \sum_{j\in \jJall} (d-R_j+C_j)^2 \frac{{h_J}}{{h_j}} = F+G+H,
\end{align}
where we expanded the square and we defined
\begin{align}
F &=  d^2 \sum_{j\in \jJall} \frac{{h_J}}{{h_j}} \qquad \;
 G=  \sum_{j\in \jJall} (C_j-R_j)^2 \frac{{h_J}}{{h_j}}  \\
  H&=   \sum_{j\in \jJall} 2d(C_j-R_j) \frac{{h_J}}{{h_j}}.
\end{align}
It is known \cite{vershik1992} that
\begin{align}
\label{eq:identityuseful}
\sum_{j\in \jJall}  \DP_j = N \DP_J \quad\quad \forall J\in {\rm Irr}(U^{\otimes N-1}),
\end{align}
which can be using Eq. (\ref{eq:defdim}) equivalently rewritten as
\begin{align}
\label{eq:identityuseful1}
  \sum_{j\in \jJall}  \frac{h_J}{h_j} = 1 \quad\quad \forall J\in {\rm Irr}(U^{\otimes N-1}).
\end{align}
Using the identity (\ref{eq:identityuseful1}) we have that
$F = d^2$.
Moreover, we have
\begin{align}
  \sum_{j\in \jJall} (d-R_j+C_j) \frac{{h_J}}{{h_j}} =
  \sum_{j\in \jJall} \frac{{l_j}}{{l_J}}  \frac{{h_J}}{{h_j}}  =
  \sum_{j\in \mathsf{j}_{JJ}} \frac{{d_j}}{{d_J}}  = d,
\end{align}
where we used Eq. (\ref{eq:crucialidentity}) and the fact that $d_j=0$ if $j$ has more than $d$ rows. 
On the other hand
\begin{align}
  \sum_{j\in \jJall} (d-R_j+C_j) \frac{{h_J}}{{h_j}} &=
  d\sum_{j\in \jJall}  \frac{{h_J}}{{h_j}}  + \sum_{j\in \jJall} (C_j-R_j) \frac{{h_J}}{{h_j}} \\
  &=  d + \frac{1}{2d}H
\end{align}
and then $H=0$.
Combining the above considerations equation (\ref{eq:defFGH}) reads
\begin{align}
\label{eq:LHSeval1}
  \sum_{j\in \mathsf{j}_{JJ}} \frac{{l_j}^2}{{l_J}^2}  \frac{{h_J}}{{h_j}}  =
  d^2 + \sum_{j\in \jJall} (C_j-R_j)^2 \frac{{h_J}}{{h_j}}.
\end{align}
Comparing Eq. (\ref{eq:LHSeval1}) with Eq. (\ref{eq:crucial}) we conclude we still need to prove
\begin{align}
  \sum_{j\in \jJall} (C_j-R_j)^2 \frac{h_J}{h_j} = N-1\;.
\end{align}
Luckily, the above equation is exactly the claim of Theorem \ref{thm:combthm1} written using Eq. (\ref{eq:defdim}). Thus, relaying on Theorem \ref{thm:combthm1} we conclude the proof.
\end{proof}

Let us continue with the proof of Theorem \ref{thm:Main}.
We multiply inequality (\ref{eq:20}) for every $J\in{\rm Irr}(U^{\otimes N-1})$ and every $j\in \mathsf{j}_{JJ}$ by $f(j,J)$. We sum these inequalities and thanks to Eqs. (\ref{eq:sumJfjJ}), (\ref{eq:optstate}) we get
\begin{align}
\label{eq:sumoffjJ}
\sum_{J\in{\rm Irr}(U^{\otimes N-1})} \sum_{j\in \mathsf{j}_{JJ}}& f(j,J)\; d_j^2 d_J \mu_J \nonumber\\
& \leq
\sum_{j \in {\rm Irr}(U^{\otimes N})} \sum_{J\in\usbl_j} f(j,J)p_j \nonumber\\
&\leq \sum_{j \in {\rm Irr}(U^{\otimes N})} p_j=1\;.
\end{align}
Let us define
\begin{align}
\label{eq:genzJ}
z_J &\equiv d_J \sum_{j\in \mathsf{j}_{JJ}} f(j,J)\; d_j^2 \quad \quad \forall J\in {\rm Irr}(U^{\otimes N-1}), \nonumber\\
&=d_J\; \DP_J \sum_{j\in \mathsf{j}_{JJ}}\frac{d_j^2}{\DP_j}=\frac{N-1 + d^2}{N}\;\;d_J^3
\end{align}
where we used Eq. (\ref{eq:deffjJ}) and Lemma \ref{lmm:identity42groups}.
Using the definition (\ref{eq:genzJ}) we can rewrite
inequality (\ref{eq:sumoffjJ}) as
\begin{align}
\label{eq:genup1}
\sum_{J\in{\rm Irr}(U^{\otimes N-1})} z_J \mu_{J}  \leq 1.
\end{align}
We remind that ${\rm Irr}(U^{\otimes N-1})=\usbl$. 
Taking this into account inequality (\ref{eq:genup1}) directly implies
\begin{align}
\label{eq:mainub1}
\frac{N-1+d^2}{N}\sum_{J\in \usbl} d_J^3 \mu_J \leq 1 \; \Leftrightarrow \;
\lambda=\sum_{J\in \usbl} \lambda_{J}\leq \frac{N}{N-1+d^2}.
\end{align}
Next, we finish the proof of Theorem \ref{thm:Main} by showing that the upper bound (\ref{eq:mainub1}) can be saturated. One can choose
\begin{align}
\label{eq:mainpjmuJ}
p_j&=\frac{d_j^2}{\sum_{k\in {\rm Irr}(U^{\otimes N})} d_k^2 } & \forall j\in {\rm Irr}(U^{\otimes N})\nonumber \\
\mu_J&=\frac{1}{d_J}\frac{1}{\sum_{k\in {\rm Irr}(U^{\otimes N})} d_k^2 } & \forall J\in\usbl 
\end{align}
and insert them into Eqs. (\ref{eq:optimizationpjmj}). It is easy to see that requirements on $p_j$ are satisfied and inequalities between $p_j$ and $\mu_J$ are actually all saturated. Let us now evaluate $\lambda$. Inserting Eq. (\ref{eq:mainpjmuJ}) into Eq. (\ref{eq:optimizationpjmj}) we obtain
\begin{align}
\label{eq:evallambda1}
\lambda 
 &=\frac{\sum_{J\in\usbl} d_J^2}{\sum_{k\in {\rm Irr}(U^{\otimes N})} d_k^2 } \nonumber \\
 &=\frac{N}{N-1 + d^2}\frac{1}{\sum_{k\in {\rm Irr}(U^{\otimes N})} d_k^2 }
 \sum_{J\in\usbl} \sum_{j\in \mathsf{j}_{JJ}} d_j^2 \frac{\DP_J}{\DP_j}       \nonumber \\
 &=\frac{N}{N-1 + d^2}\frac{1}{\sum_{k\in {\rm Irr}(U^{\otimes N})} d_k^2 }
  \sum_{j\in {\rm Irr}(U^{\otimes N})} d_j^2 \sum_{J\in\usbl_j} \frac{\DP_J}{\DP_j}  \nonumber \\
 &=\frac{N}{N-1 + d^2}\frac{1}{\sum_{k\in {\rm Irr}(U^{\otimes N})} d_k^2 }
  \sum_{j\in {\rm Irr}(U^{\otimes N})} d_j^2                                  \nonumber \\
 &=\frac{N}{N-1 + d^2}\; ,
\end{align}
where we used Lemma \ref{lmm:identity42groups}, exchanged the order of the sums and used the Eq. (\ref{eq:sumJfjJ}). Thanks to knowledge of $\mu_J$ and $p_j$ we can completely specify the state $\ket{\psi}$ and the retrieving operation $\mathcal{R}_s$. Thus, we can build valid storing and retrieving strategy, which succeeds with probability $N/(N-1 + d^2)$ saturating the upper bound (\ref{eq:mainub1}) and concluding the proof.
\end{proof}

\section{Alignment of reference frames}

We now review the quantum protocol for the alignment of reference
frames in a quantum communication scenario, as it was considered in
Ref. \cite{refframes}.  Let us consider the scenario in which one
party, called Alice, wants to send a qubit to another distant party, denoted as Bob. If
the qubit is encoded into a spin-$1/2$ particle Bob can recover the
quantum state $|\varphi\>$ if he and Alice share a reference frame for
orientation.  Otherwise, the lack of a shared frame amounts to having
a noisy channel and Bob receives a decohered state
$\rho = \int_{SU(2)} U |\varphi\>\<\varphi|U^\dag d U$.  This problem
can be circumvented if Alice, along with the quantum message
$|\varphi\> $, sends a state $|\psi\> $ as a token of her reference
frame.  Then Bob receives the state
$\rho_\psi = \int_{SU(2)} |\psi_U\>\<\psi_U| \otimes U
|\varphi\>\<\varphi|U^\dag d U$, from which he tries to retrieve the
message $|\varphi\>$. In the perfect retrieving scenario, Bob wants to
maximize the probability for recovering $|\varphi\>$ without any
error.  This scenario is equivalent to a storage and
retrieval protocol:
\begin{itemize}
\item the token $\ket{\psi}$ plays the role of the storage state.
  $\ket{\psi}$ is a multipartite state $\ket{\psi} \in
  \hilb{H}^{\otimes N}$
\item
  The effect of the misalignment can be thought of as the storing
  phase in which the state $\ket{\psi_U} := U^{\otimes N} \ket{\psi}$
  is created.
\item
  In the retrieving phase, Bob exploits the state
  $\ket{\psi_U}$ to retrieve the inverted channel $U^\dag$
  which is applied to the qubit $U |\varphi\>$.
\end{itemize}

There are two differences between this protocol and the SAR we
consider in our work. The first
difference is that $N$ uses of $U$ are given, but we are required to
retrieve $U^\dag$.
However, for $U \in SU(2)$, we can show that our optimal PSAR  protocol
which stores $U^{\otimes N}$ and retrieves $U$,
can be turned into a PSAR protocol which retrieves $U^{\dag}$
with the same probability of success.
If we had $\ket{\psi_{U^\dag}}$, then the retrieval phase of our
optimal PSAR protocol would
recover $U^{\dag}$ with the optimal probability of success $\lambda$
(which is the same for any $U \in SU(2)$).
In particular, for $U \in SU(2)$, the storage state
$\ket{\psi_{U^\dag}}$ can be created by exploiting $N$ uses of $U$ as
follows
\begin{align}
  \begin{split}
     \ket{\psi_{U^\dag}} &= U^{ \dag \otimes N} \otimes I \ket{\psi}
  =I \otimes    U^{ * \otimes N} \ket{\psi} =\\
&=  I \otimes  (\sigma_y U \sigma_y)^{\otimes N}\ket{\psi}
  \end{split}
\end{align}
where $\ket{\psi}$ is the optimal state for storage.

The second difference between SAR and the alignment protocol is that
we are not allowed to use an external reference system,
i.e. the ancillary system $\hilb{H}_{A'}$ in our protocol,
since it would correspond to a partially shared reference frame.
Since our protocol is less constraint than the alignment protocol,
the probability of success $\lambda$ of PSAR is an upper bound for the
probability of success of perfect alignment.
However both the strategy of Ref. \cite{refframes} and the optimal
PSAR protocol achieve the same $O(N^{-1})$ scaling, which is then optimal.


\begin{thebibliography}{99}
\bibitem{foot1}
In the literature, the task of storage and retrieval is also called quantum learning.
\end{thebibliography}
\end{document}